\documentclass{article}

\usepackage{arxiv}

\usepackage{amsthm}

\usepackage{xcolor}  % Needed for colours - load before mdframed

\usepackage{amssymb}
\usepackage{amsmath}
\usepackage{mathdots}

\usepackage{mathtools}

\usepackage{tensor}

\usepackage{urwchancal}
\DeclareFontFamily{OT1}{pzc}{}
\DeclareFontShape{OT1}{pzc}{m}{it}{<-> s * [1.10] pzcmi7t}{}
\DeclareMathAlphabet{\mathpzc}{OT1}{pzc}{m}{it}

\usepackage{graphicx}
\usepackage{ifpdf}
\ifpdf
\usepackage{epstopdf}
\PrependGraphicsExtensions{.svg}
\fi
\newtheorem{theorem}{Theorem}[section]
\newtheorem{lemma}[theorem]{Lemma}

\newtheorem{proposition}[theorem]{Proposition}

\providecommand{\R}{\mathbb{R}}

\providecommand{\SO}{\mathbf{SO}}

\providecommand{\MR}{\mathbf{MR}}

\providecommand{\grpG}{\mathbf{G}}

\providecommand{\grpS}{\mathbf{S}}
\providecommand{\grpN}{\mathbf{N}}

\providecommand{\gothg}{\mathfrak{g}}
\providecommand{\gothh}{\mathfrak{h}}

\providecommand{\gothX}{\mathfrak{X}} % as in X(M)

\providecommand{\calD}{\mathcal{D}}

\providecommand{\calF}{\mathcal{F}}

\providecommand{\calM}{\mathcal{M}}
\providecommand{\calN}{\mathcal{N}}

\providecommand{\vecV}{\mathbb{V}}

\providecommand{\DM}{\mathcal{D}(\mathcal{M})}

\providecommand{\I}{\mathrm{I}}
\providecommand{\pt}{\mathrm{P}}

\DeclareMathOperator{\Ad}{Ad}

\providecommand{\id}{\mathrm{id}} % identity map
\providecommand{\tT}{\mathrm{T}} % tangent bundles

\providecommand{\tD}{\mathrm{D}}
\providecommand{\tL}{\mathrm{L}}
\providecommand{\tR}{\mathrm{R}}

\providecommand{\mr}[1]{{#1}^\circ} % reference element.

\usepackage{accents}
\makeatletter
\providecommand{\scirc}{%
    \hbox{\fontfamily{\rmdefault}\fontsize{0.4\dimexpr(\f@size pt)}{0}\selectfont{\raisebox{-0.52ex}[0ex][-0.52ex]{$\circ$}}}}
\makeatother
\mathchardef\mhyphen="2D
\providecommand{\etal}{\textit{et al.}~}

\usepackage[all]{xy}
\usepackage{hyperref}
\usepackage{stfloats}

\newcommand{\papercitation}{% add citation to published paper in text
Ge, Y., van Goor, P., Mahony, R. (2022), ``Equivariant Filter Design for Discrete-time Systems'', \emph{Accepted for presentation at IEEE CDC}.
}

\newcommand{\publicationdetails}{
\papercitation \\
\copyrightNoticeIEEE{2022} % Use for IEEE papers
}
\newcommand{\publicationversion}
{Author accepted version}

\begin{document}
\title{
Equivariant Filter Design for Discrete-time Systems
}
\headertitle{
Equivariant Filter Design for Discrete-time Systems
}
\author{
    \href{https://orcid.org/0000-0001-7969-7039}{\includegraphics[scale=0.06]{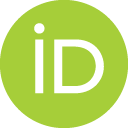}\hspace{1mm}
Yixiao Ge}
\\
    Systems Theory and Robotics Group \\
    Australian Centre for Robotic Vision \\
	Australian National University \\
    ACT, 2601, Australia \\
    \texttt{Yixiao.Ge@anu.edu.au} \\
\And    \href{https://orcid.org/0000-0003-4391-7014}{\includegraphics[scale=0.06]{orcid.png}\hspace{1mm}
Pieter van Goor}
\\
    Systems Theory and Robotics Group \\
    Australian Centre for Robotic Vision \\
	Australian National University \\
    ACT, 2601, Australia \\
    \texttt{Pieter.vanGoor@anu.edu.au} \\
    %%%%%%%%%% - comment removed - %%%%%%%%%%
	\And	\href{https://orcid.org/0000-0002-7803-2868}{\includegraphics[scale=0.06]{orcid.png}\hspace{1mm}
    Robert Mahony}
\\
    Systems Theory and Robotics Group \\
    Australian Centre for Robotic Vision \\
	Australian National University \\
    ACT, 2601, Australia \\
	\texttt{Robert.Mahony@anu.edu.au} \\
}

\maketitle
% \thispagestyle{empty}
% \pagestyle{empty}

%%%%%%%%%%%%%%%%%%%%%%%%%%%%%%%%%%%%%%%%
\begin{abstract}
The kinematics of many nonlinear control systems, especially in the robotics field, admit a transitive Lie-group symmetry, which is useful in high performance observer design.
The recently proposed equivariant filter (EqF) exploits equivariance to generate high performance filters for a wide range of real-world systems.
However, existing work on the equivariant filter, and equivariance of control systems in general, is based on a continuous-time formulation.
In this paper, we first present the equivariant structure of a discrete-time system.
We then use this to propose a discrete-time version of the equivariant filter.
A novelty of the approach is that the geometry of the symmetry group naturally appears as parallel transport in the reset step of the filter.
Preliminary results for linear second order kinematics with separate bearing and range measurements indicate that the discrete EqF significantly outperforms both a discretized version of the continuous EqF and a  classical discrete EKF.
\end{abstract}

%%%%%%%%%%%%%%%%%%%%%%%%%%%%%%%%%%%%%%%%
\section{INTRODUCTION}
%%%%%%%%%%%%%%%%%%%%%%%%%%%%%%%%%%%%%%%%
Many nonlinear control systems are naturally posed on homogeneous spaces, that is manifolds that submit to a transitive group action by a Lie group.
Such systems carry the property of equivariance, that is, the defining dynamics can be mapped throughout state-space by the group action.
The geometry and symmetry of those systems have been recognized since the seventies \cite{FB_ADL_04}, \cite{Jurdjevic_1997}, although this did not translate into observer and filter design until the 1990s.
Salcudean proposed a non-linear observer for the attitude estimation of a satellite \cite{Salcudean_1991}.
Thienel and Sanner extended the original observer and incorporated bias estimation \cite{Thienel_2003}.
Parallel work from Aghannan and Rouchon introduced a state observer design for Lagrangian systems using the invariance property \cite{Aghannan_2003}.
Later around 2005, with the emerging unmanned aerial vehicle industry, there was a burst of work to develop a simple, robust observer design for attitude control.
Mahony \etal \cite{Mahony_2008} proposed an asymptotically stable complementary filter for attitude estimation posed directly on the matrix Lie group $\SO(3)$, which has an almost-global convergence property.
In parallel, Bonnabel \etal proposed a general theory for the Invariant Extended Kalman Filter (IEKF) in a series of works \cite{Bonnabel_2007}, \cite{Bonnabel_2009}.
They provided a design methodology for systems on Lie groups with invariance properties, which was later extended to the class of `group affine' systems \cite{Barrau_2017}.
These works fully exploited the underlying symmetry properties of the system and had great impact in the control and robotics community.
Recently, van Goor \etal \cite{vanGoor_Hamel_Mahony_2021} and  Mahony \etal \cite{mahony2020equivariant} proposed the Equivariant Filter (EqF), a general filter design for systems on homogeneous space.

Less attention has been focused on the observer problem for discrete-time systems.
There are two common approaches to design a discrete filter for a continuous system with discrete-time measurements.
The first is to design a continuous-time filter and then discretize the filter, while the second approach is to discretize the continuous system and design a discrete-time filter for the discrete system \cite{Chung_1990}.
One of the advantages of the second is a range of existing geometric integrators to discretize systems on manifolds, provided by Crouch \etal and Hairer \etal \cite{crouch1993numerical}\cite{hairer2006geometric}.
Nonetheless, both approaches introduce approximation error into the system with the discretization process, although, the second is known to have better stability properties \cite{Hori_1989}.
Preliminary work on stochastic filtering of continuous- and discrete-time systems was carried out in the late 1970s and early 1980s, and is summarised in Maybeck's book \cite{maybeck1982stochastic}.
More recently, the question of discrete filtering for systems on Lie-groups has received considerable attention.
Bourmaud \etal generalized the extended Kalman filter to matrix Lie groups for discrete-time systems \cite{Bourmaud_2015}.
Barrau \etal proposed the IEKF for discrete-time dynamical systems on Lie groups, which has been implemented for navigation problems \cite{barrau2018invariant}.

In this paper, we propose an observer design methodology for a class of discrete-time systems on homogeneous spaces using the equivariance property.
This work is parallel to recent work on Equivariant filter (EqF) design \cite{vanGoor_Hamel_Mahony_2021}, that has yielded filters that are obtained as a discretization of a continuous-time filter design. 
We assume that a discretization of the continuous-time system is provided as the model, exploiting for example the literature on discretization of systems on Lie-groups \cite{crouch1993numerical}\cite{hairer2006geometric}\cite{forster2016manifold}\cite{barrau2019linear}. 
In particular, these discretizations yield updates in the form of a symmetry transformation that approximates the integration of the continuous-time system for constant input acting on the present state. 
The symmetry structure is exploited in the filter development and we believe lies at the heart of the performance gains reported.

One of our key contributions is that we propose a novel approach to model filter reset as the parallel transport of covariance.
In a conventional extended Kalman filter, the update step provides an estimate of the new mean and covariance based on the predicted state and new measurements.
We implement Bayesian fusion in local coordinates centred at the predicted state estimate generating both a new state estimate and covariance expressed in the local coordinates, and then, drawing from recent insight \cite{doi:10.1146/annurev-control-061520-010324}, parallel transport the covariance to new coordinates centred at the updated state estimate.
Our approach provides an alternative, and we claim more general, perspective on work done by Markley in \cite{markley2003attitude} twenty years ago and generalised recently by Mueller \etal \cite{mueller2017covariance}.
We apply the proposed filter to an example of second-order kinematics with separate range and bearing measurements and run the comparison with a discretized version of a continuous-time EqF and conventional discrete extended Kalman filter.
The simulations demonstrate the advantage of both EqF filters versus the EKF due to proper modelling of the reset step during the transient response, and significant performance gains of the discrete EqF compared to both comparison filters in the asymptotic response.
We believe this performance gain is due to the fact that the error incurred in a discrete state update implemented as a symmetry transform is better adapted to the discrete equivariant filter equations than the error incurred by the discretization of a continuous-time filter, even when the associated filter was based on equivariant design principles in continuous-time.

%%%%%%%%%%%%%%%%%%%%%%%%%%%%%%%%%%%%%%%%
\section{PRELIMINARIES}
\label{sec:notation}
%%%%%%%%%%%%%%%%%%%%%%%%%%%%%%%%%%%%%%%%
Let $\calM$ be a smooth manifold.
The tangent space at a point $\xi\in\calM$ is denoted $\mathrm{T}_\xi\calM$.
A diffeomorphism on $\calM$ is a smooth map $F:\calM\rightarrow\calM$.

Let $\gothX(\calM)$ denote the space of vector fields over $\calM$.
Let $\calD(\calM)$ denote the group (under function concatenation) of diffeomorphisms of $\calM$.
Let $f \in \gothX(\calM)$ then the flow $\calF: \R\times\calM\rightarrow\calM$ of $f$ is defined as the family of diffeomorphisms $\calF(t,\cdot) \in \calD(\calM)$, $\calF(t,\xi_0) = \xi(t)$, where $\xi(t)$ is the solution of $\dot{\xi} = f(\xi)$ for $\xi(0) = \xi_0$.

Given a differentiable function between smooth manifolds $h:\calM\rightarrow\calN$, its derivative at $\xi_\circ$ is written as
\begin{align*}
    \tD_{\xi|\xi_\circ}h(\xi): \tT_{\xi_\circ}\calM\rightarrow \tT_{\xi_\circ}\calN.
\end{align*}
The notation $\tD h(\xi):\tT\calM\rightarrow \tT\calN$ denotes the differential of $h$ with an implicit base point.

Let $\grpG$ be a general Lie group with Lie algebra $\gothg$, $\id$ denotes the identity element of $\grpG$.
For arbitrary $X, Y\in\grpG$, the left and right translations are denoted $\tL_X$ and $\tR_X$, and are defined by
\begin{equation*}
    \tL_X(Y):=XY, \quad \tR_X(Y):=YX.
\end{equation*}
The Lie algebra $\gothg$ is isomorphic to a vector space $\R^n$ with the same dimension.
We use the wedge $(\cdot)^\wedge:\R^n\rightarrow\gothg$ and vee $(\cdot)^\vee:\gothg\rightarrow\R^n$ operators to map between the Lie algebra and vector space.
The adjoint map for $\grpG$, $\Ad_X:\gothg\rightarrow\gothg$ is defined by
\begin{align*}
    \Ad_XU = \tD \tL_X \tD \tR_{X^{-1}}U
\end{align*}
for any $X\in\grpG$ and $U\in\gothg$.
When $\grpG$ is a matrix Lie group, we can define the adjoint matrix $\Ad^\vee_X:\R^n \rightarrow \R^n$ as
\begin{align*}
    \Ad^\vee_XU = (\Ad_XU)^\vee
\end{align*}
For $X\in\grpG$ the \emph{inner automorphism} $\I_X: \grpG\rightarrow\grpG$
of $\grpG$ is the smooth map
\begin{equation*}
    \I_X(Z):=XZX^{-1}=\tL_X\tR_{X^{-1}}Z, \quad Z\in\grpG.
\end{equation*}

A right group action of a Lie group $\grpG$ on a manifold $\calM$ is a smooth map $\phi:\grpG\times\calM\rightarrow\calM$ that satisfies
\begin{align*}
    \phi(X,\phi(Y,\xi))=\phi(YX,\xi) \quad \textrm{and} \quad \phi(\id,\xi)=\xi
\end{align*}
for any $X,Y \in\grpG$ and $\xi\in \calM$.
It induces the partial maps $\phi_X:\calM\rightarrow\calM$ and $\phi_\xi:\grpG\rightarrow\calM$ which are defined by $\phi_X(\xi):=\phi(X,\xi) =: \phi_\xi(X)$ respectively.

\begin{proposition}
    Any right action $\phi: \grpG\times\calM \rightarrow \calM$ induces a right action on the group of diffeomorphisms of $\calM$, denoted $\Phi: \grpG\times\calD(\calM)\rightarrow\calD(\calM)$, and defined by
    \begin{equation}\label{eq:Phi}
        \Phi(X,F) = \phi_X\circ F \circ \phi_X^{-1},
    \end{equation}
    for any $F\in\DM$.
\end{proposition}
\begin{proof}
 Let $X, Y \in \grpG$ and $F\in\calD(\calM)$.
    To prove compatibility, compute
    \begin{align*}
        \Phi_Y(\Phi_X(F)) &= \Phi_Y(\phi_XF\phi_X^{-1})\\
        &=\phi_Y(\phi_XF\phi_X^{-1})\phi_Y^{-1}\\
        &=\phi_{XY}F\phi_{(XY)^{-1}} = \Phi_{XY}F.
    \end{align*}
    To prove the identity,
    \begin{align*}
        \Phi_\id F = \phi_\id F\phi_\id = F,
    \end{align*}
    as required.
\end{proof}

%%%%%%%%%%%%%%%%%%%%%%%%%%%%%%%%%%%%%%%%
\section{PROBLEM FORMULATION}
\label{sec:formulation}
%%%%%%%%%%%%%%%%%%%%%%%%%%%%%%%%%%%%%%%%
\subsection{Discrete-time System}
Let $\calM$ be a smooth manifold termed the \emph{state space}, $\vecV$ be a vector space termed the \emph{input space}.
Let $F : \calM \times \vecV \to \calM$ be a family of diffeomorphisms on $\calM$, and $F(\cdot, u_k) \in \calD(\calM)$ for all $u_k \in \vecV$.
The class of discrete-time systems on $\calM$ that we consider can be written as
\begin{equation}
    \label{eq:evo}
    \xi_{k+1} = F_{u_{k+1}}(\xi_k) = F(\xi_k,u_{k+1}),\quad \xi_0\in \calM,
\end{equation}
The configuration output $y_k = h(\xi_k)$ is given by a function $h: \calM\rightarrow\calN\subset\R^n$, where $\calN$ is a smooth manifold termed the \emph{output space} embedded in $\R^n$.
In practice, there will be disturbances on both the update and the measurements process. 
The noise process for the measurement is a Gaussian modelled in the embedding space $\R^n$, while for the process noise, we model a Gaussian disturbance in the local exponential coordinates during the update step (the $P$ matrix in \eqref{eq:sigma_predict}). 

\subsection{State Symmetry}
Let $\grpG$ be a Lie group with Lie algebra $\gothg$, and $\calM$ be a homogeneous space of $\grpG$, then there exists a smooth, transitive, right group action of $\grpG$ on $\calM$,
\begin{equation}
    \phi: \grpG\times\calM \rightarrow \calM.
\end{equation}

A \emph{discrete lift} for the system (\ref{eq:evo}) is a map $\Lambda: \calM\times\vecV\rightarrow\grpG$ with the lift condition
\begin{equation}
    \label{eq:lift_condition}
    \phi_{\xi_k}(\Lambda(\xi_k, u_{k+1})) = F_{u_{k+1}}(\xi_k),
\end{equation}
for every $\xi\in\calM$ and $u\in\vecV$.
The lift $\Lambda: \calM\times\vecV \rightarrow \grpG$ is used to define a \emph{lifted system} on the symmetry Lie group $\grpG$.
First choose an arbitrary origin $\xi^\circ\in\calM$, that acts as the origin in a global coordinate parametrization $\phi_{\xi^\circ} : \grpG \rightarrow \calM$.
The lifted system is written
\begin{equation}
    \label{eq:lift}
    X_{k+1} = X_k\Lambda(\phi_{\xi^\circ}(X_k),u_{k+1}), \quad X(0)\in\grpG.
\end{equation}

\begin{lemma}
Consider the lifted system \eqref{eq:lift}.
If the initial condition $X(0)\in\grpG$ satisfies
\[
\phi_{\xi^\circ}(X(0))=\xi(0),
\]
then the trajectory of the lifted system $X_k$ projects down to the original kinematic system with the map $\phi_{\xi^\circ}: \grpG\rightarrow\calM$
\begin{equation}
    \phi_{\xi^\circ}(X) \equiv \xi,
\end{equation}
for all time.
\end{lemma}

\begin{proof}
Let $X_k$  be a solution of (\ref{eq:lift}) and let $\xi_k = \phi_{\xi^\circ}(X_k)$.
One has
\begin{align*}
    \xi_{k+1} = \phi_{\xi^\circ}(X_{k+1}) &= \phi_{\xi^\circ}(X_k\Lambda(\phi_{X_k}(\xi^\circ),u_{k+1}))\\
    % &= \phi(X_k\Lambda(\xi_k,u_{k+1}),\xi^\circ)\\
    &= \phi(\Lambda(\xi_k,u_{k+1}),\phi(X_k,\xi^\circ))\\
    &= \phi_{\xi_k}(\Lambda(\xi_k, u_{k+1}))\\
    &= F(\xi_k,u_{k+1}),
\end{align*}
which satisfies (\ref{eq:evo}).
\end{proof}

%%%%%%%%%%%%%%%%%%%%%%%%%%%%%%%%%%%%%%%%
\section{EQUIVARIANT SYSTEMS}
\label{sec:symmetry}
%%%%%%%%%%%%%%%%%%%%%%%%%%%%%%%%%%%%%%%%
\subsection{Equivariance}
A discrete-time system (\ref{eq:evo}) is termed \textit{equivariant} if there exists a group action
\begin{align}
    \psi:\grpG\times\vecV\rightarrow\vecV
\end{align}
which satisfies
\begin{equation}
    \label{eq:equivariance}
    \phi_X (F(\xi,u))=F(\phi_X(\xi), \psi_X(u)),
\end{equation}
for any $X\in\grpG$, $\xi\in\calM$ and $u\in\vecV$.
Consider a discrete-time equivariant system with the group action, then
\begin{align}
        F_{\psi_X(u)}(\xi) &= F_{\psi_X(u)}(\phi_X(\phi_X^{-1}(\xi)))\nonumber\\
        &= \phi_X(F_u(\phi_X^{-1}(\xi))\nonumber\\
        &= \Phi_XF_u(\xi), \label{eq:input}
\end{align}
where $\Phi$ is the group action on $\calD(\calM)$ \eqref{eq:Phi}.
We can also conclude from  (\ref{eq:input}) that if $\psi$ exists, i.e. the system \eqref{eq:evo} is equivariant, then $\psi$ is implicitly determined by the induced group action $\Phi$.

\subsection{Equivariant Lift}
For a system with state symmetry $\phi: \grpG\times\calM \rightarrow \calM$ and input symmetry $\psi: \grpG\times\vecV\rightarrow\vecV$, the discrete lift is \emph{equivariant} if
\begin{equation}
    \label{eq:eqlift}
    \Lambda(\phi(X,\xi), \psi(X, u)) = X^{-1}\Lambda(\xi, u)X = \I_{X^{-1}}\Lambda(\xi, u) ,
\end{equation}
for any $X\in\grpG$, $\xi\in\calM$ and $u\in\vecV$.

%%%%%%%%%%%%%%%%%%%%%%%%%%%%%%%%%%%%%%%%
\section{OBSERVER ARCHITECTURE}
\label{sec:observer}
%%%%%%%%%%%%%%%%%%%%%%%%%%%%%%%%%%%%%%%%
In this section we provide the formulation of an equivariant observer for a discrete-time system.
The key in the formulation is the equivariant error which measures the global error between the manifold and the symmetry group.
% The approach is based on the extended Kalman Filter applied to the equivariant error.
\subsection{Observer Formulation}
Let $\calM$ be a smooth manifold.
Consider a kinematic model with state symmetry $\phi: \grpG\times\calM \rightarrow \calM$, input symmetry $\psi: \grpG\times\vecV\rightarrow\vecV$ and trajectory denoted by $\xi_k$.
Choose an arbitrary but fixed origin $\xi^\circ\in\calM$.
The lifted system is defined by the evolution function
\begin{equation}
    \label{eq:liftedsystem}
    {X}_{k+1} = {X}_k\Lambda(\phi_{\xi^\circ}({X}_k),u_{k+1}), \quad \phi(X(0),\xi^\circ)=\xi(0)
\end{equation}
where $X(0)\in\grpG$.

Define the observer $\hat{X}_k$ dynamics on the symmetry Lie group by
\begin{equation}
    \label{eq:observer}
    \hat{X}_{k+1} = \exp(\Delta_{k+1})\hat{X}_k\Lambda(\phi_{\xi^\circ}(\hat{X}_k),u_{k+1}), \quad \hat{X}(0)=\id
\end{equation}
where $\Delta_{k+1}\in\gothg$ is the time varying correction term that remains to be chosen.
The corresponding state estimate of the observer is given by
\begin{equation*}
    \hat{\xi}_k=\phi_{\xi^\circ}(\hat{X}_k).
\end{equation*}
The \textit{equivariant error} is defined as
\begin{equation}
    \label{eq:equierror}
    e_{k}=\phi_{\hat{X}_k^{-1}}(\xi_{k}).
\end{equation}
This error is a key concept in equivariant filter design, since it measures the difference between the observer trajectory $\hat{X}_k\in\grpG$ and the system trajectory $\xi_k\in\calM$ which live in two different spaces.
Note that if $\hat{\xi_k}=\xi_k$, i.e the observer estimate is correct, then the error $e_k=\xi^\circ$.

Define the \textit{origin input}
\begin{align}
    \label{eq:orivel}
    u^\circ_{k+1}:=\psi(\hat{X}_k^{-1},u_{k+1}),
\end{align}
which is measurable because both $\hat{X}$ and $u$ are known.

\subsection{Error Dynamics}
The EqF is derived by linearising the dynamics of the equivariant error $e_k$.
\begin{lemma}
    Consider an equivariant system with lifted model (\ref{eq:lift}) and the observer (\ref{eq:observer}), the error dynamics are given by
    \begin{equation}
        \label{eq:error_dynamics}
        e_{k+1}=\phi(\Lambda(e_k,u^\circ_{k+1})\Lambda(\xi^\circ,u^\circ_{k+1})^{-1}\exp(-\Delta_{k+1}), e_k),
    \end{equation}
    and only depend on the origin input ${u}^\circ_{k+1} = \psi(\hat{X}_k^{-1},u_{k+1})$ and the error $e_k$ itself.
\end{lemma}
\begin{proof}
    Let $X_k$ be a solution to the lifted model (\ref{eq:lift}) satisfying $\phi(X(0),\xi^\circ)=\xi(0)$.
    Define the group error $E_k:=X_k\hat{X}_k^{-1}$, compute the dynamics of $E$:
{\small
    \begin{align*}
 E_{k+1} & = X_k\Lambda(\phi_{\xi^\circ}(X_k),u_{k+1})\Lambda(\phi_{\xi^\circ}(\hat{X}_k),u_{k+1})^{-1}\hat{X}_k^{-1}\exp(-\Delta_{k+1}),\\
    &=E_k\hat{X}_k\Lambda(\phi_{\xi^\circ}(E_k\hat{X}_k),u_{k+1})\hat{X}_k^{-1}\hat{X}_k\Lambda(\phi_{\xi^\circ}(\hat{X}_k),u_{k+1})^{-1}\\
    &\qquad\qquad\qquad\qquad\qquad\qquad\qquad\hat{X}_k^{-1}\exp(-\Delta_{k+1}),\\
    &=E_k\Lambda(\phi_{\xi^\circ}(E_k),\psi(\hat{X}_k^{-1},u_{k+1}))\Lambda(\phi_{\xi^\circ}(I),\psi(\hat{X}_k^{-1},u_{k+1}))^{-1}\\
    &\qquad\qquad\qquad\qquad\qquad\qquad\qquad\exp(-\Delta_{k+1}),\\
    &=E_k\Lambda(\phi(E_k,\xi^\circ),u^\circ_{k+1})\Lambda(\xi^\circ,u^\circ_{k+1})^{-1}\exp(-\Delta_{k+1}),\\
    &=E_k\Lambda(e_k,u^\circ_{k+1})\Lambda(\xi^\circ,u^\circ_{k+1})^{-1}\exp(-\Delta_{k+1}).
    \end{align*}
}
    Note that
    \begin{align*}
        e_k: = \phi(\hat{X_k}^{-1}, \xi_k) = \phi(X_k\hat{X_k}^{-1}, \xi^\circ) = \phi(E_k, \xi^\circ).
    \end{align*}
    Substitute the dynamics of $E_k$ into this relationship and one obtains
    \begin{align*}
        e_{k+1}&=\phi(E_{k+1},\xi^\circ),\\
    &=\phi(E_k\Lambda(e_k,u^\circ_{k+1})\Lambda(\xi^\circ,u^\circ_{k+1})^{-1}\exp(-\Delta_{k+1}),\xi^\circ),\\
    &=\phi(\Lambda(e_k,u^\circ_{k+1})\Lambda(\xi^\circ,u^\circ_{k+1})^{-1}\exp(-\Delta_{k+1}),\phi(E_k,\xi^\circ)),\\
    &=\phi_{e_k}(\Lambda(e_k,u^\circ_{k+1})\Lambda(\xi^\circ,u^\circ_{k+1})^{-1}\exp(-\Delta_{k+1})),
    \end{align*}
    as required.
\end{proof}

\subsection{Linearisation}
To linearise the error $e_k\in\calM$ requires a chart of local coordinates for the state.
Fix a local coordinate chart $\varTheta:\calM\rightarrow\R^m$ around the origin $\mr{\xi}$ and assume that $e$ remains in a neighborhood of the origin for all time.
We use the notation $\varTheta$ and $\varepsilon$ to represent the chart and the local coordinates of the error respectively,
\begin{align}
    \varepsilon_k=\varTheta(e_k).
\end{align}

\begin{proposition}
    Let $\varTheta$ be a local coordinate chart on $\calM$ in an open neighborhood around $\mr{\xi}$.
    The linearised dynamics of $e_k$ about $\varepsilon_k=0$ and $\Delta_k=0$ are
    \small
    \begin{align}
        \varepsilon_{k+1}&\approx A_{k+1}\varepsilon_k + O(\lvert \varepsilon \rvert^2),\label{eq:linearisation}\\
        A_{k+1} &= \tD\varTheta\cdot \tD\phi_{\xi^\circ}(\id)\cdot \tD \tR_{\Lambda(\mr{\xi},u^\circ_{k+1})^{-1}}(\Lambda(\mr{\xi},\mr{u}_{k+1}))\nonumber\\
        &\qquad \quad \cdot \tD_{\xi^\circ}\Lambda(\xi^\circ, \mr{u}_{k+1})\cdot \tD\varTheta^{-1}.\label{eq:statematrix}
    \end{align}
    \normalsize
\end{proposition}
\begin{proof}
    Define the error drift term
    \begin{align*}
        \tilde{\Lambda}_{\xi^\circ}(e_k,u^\circ_{k+1})=\Lambda(e_k,u^\circ_{k+1})\Lambda(\xi^\circ,u^\circ_{k+1})^{-1}.
    \end{align*}
    The global state error $e_k$ has the dynamics
    \begin{align}
        e_{k+1}&=\phi_{e_k}(\tilde{\Lambda}_{\xi^\circ}(e_k,u^\circ_{k+1})\exp(-\Delta_{k+1})).
    \end{align}
    Then,
    \begin{equation}
        \label{eq:errorevo}
        \varepsilon_{k+1}=\varTheta(\phi_{\varepsilon^{-1}(\varepsilon_k)}(\tilde{\Lambda}_{\xi^\circ}(\varepsilon^{-1}(\varepsilon_k),u^\circ_{k+1})\exp(-\Delta_{k+1}))).
    \end{equation}
    Linearising the above equation (\ref{eq:errorevo}) about $\varepsilon_k=0$ and $\Delta_k=0$ yields
    \small
    \begin{align*}
        \varTheta(&\phi_{\varTheta^{-1}(\varepsilon_k)}(\tilde{\Lambda}_{\xi^\circ}(\varTheta^{-1}(\varepsilon_k),u^\circ_{k+1})))\\
        &\approx \varTheta(\phi_{\varTheta^{-1}(\varepsilon_k)}(\tilde{\Lambda}_{\xi^\circ}(\varTheta^{-1}(0),u^\circ_{k+1})))\\&+\tD\varTheta\cdot \tD\phi_{\xi^\circ}(\id)\cdot \tD\tilde{\Lambda}(\xi^\circ, u^\circ_{k+1})\cdot \tD\varTheta^{-1}(\varepsilon_k) + O(\lvert \varepsilon \rvert^2)\\
        &=A_{k+1}\varepsilon_k + O(\lvert \varepsilon \rvert^2),\\
        A_{k+1} &= \tD\varTheta\cdot \tD\phi_{\xi^\circ}(\id)\cdot \tD \tR_{\Lambda(\mr{\xi},u^\circ_{k+1})^{-1}}(\Lambda(\mr{\xi},\mr{u}_{k+1}))\\
        &\qquad \quad \qquad\qquad \cdot \tD_{\xi^\circ}\Lambda(\xi^\circ, \mr{u}_{k+1})\cdot \tD\varTheta^{-1},
    \end{align*}\normalsize
    as required.
\end{proof}

Now consider the global state error $e_k$ and the true system state $\xi_k$, the system output $h(\xi)\in\calN\subset\R^n$ can be written as
\begin{align}
    h(\xi_k)=h(\phi(\hat{X}_k,e_k))=h(\phi_{\hat{X}_k}(\varTheta^{-1}(\varepsilon_k))).
\end{align}
Define the output residual $\tilde{y}_k\in\R^n$ to be $\tilde{y}_k=h(\xi_k)-h(\hat{\xi}_k)$.
Linearise $\tilde{y}_k$ at $\varepsilon_k=0$,\small
\begin{align}
    \tilde{y}_k&=y_k-\hat{y}_k,\nonumber\\
    &=h(\phi_{\hat{X}_k}(\varTheta^{-1}(0)))+\tD h(\hat{\xi}_k)\cdot \tD\phi_{\hat{X}_k}(\xi^\circ)\nonumber\\
    &\qquad\qquad\qquad\qquad\cdot \tD\varTheta^{-1}(\varepsilon_k)+O(\lvert \varepsilon_k\rvert^2)-\hat{y}_k,\nonumber\\
    &=\tD h(\hat{\xi}_k)\cdot \tD\phi_{\hat{X}_k}(\xi^\circ)\cdot \tD\varTheta^{-1}(\varepsilon_k)+O(\lvert \varepsilon_k\rvert^2),\nonumber\\
    &= C_k\varepsilon_k + O(\lvert \varepsilon_k\rvert^2),\nonumber\\
    C_k&=\tD h(\hat{\xi}_k)\cdot \tD\phi_{\hat{X}_k}(\xi^\circ)\cdot \tD\varTheta^{-1}.\label{eq:linearoutput}
\end{align}\normalsize

\section{OBSERVER DESIGN}\label{subsec:eqf}
In this section, we use the concept of concentrated Gaussian distributions and pose the distributions with local coordinates.
Then we derive the EqF equations for predict, update and reset steps separately, which is based on the extended Kalman Filter applied to the equivariant error.

\subsection{Concentrated Gaussians on homogeneous spaces}
We extend the concept of \emph{concentrated Gaussian distribution on Lie groups} \cite{wolfe2011bayesian},\cite{Bourmaud_2015} to homogeneous spaces and introduce the \emph{extended concentrated Gaussian distribution}, which is a generalization of the normal distribution on the Euclidean space.

Let $\xi\in\calM$ be a random variable on the manifold.
Fix an arbitrary origin on the manifold $\xi^\circ\in\calM$.
Fix $\gothh$ to be a horizontal subspace of $\gothg$.
Define a local map $\phi_{\xi^\circ}\cdot\exp_\gothh :\gothh \rightarrow \calM$.
Since $\gothh$ is horizontal, this map is always locally well defined and injective.
Now define a local coordinate chart $\varTheta_\gothh:\calM\rightarrow\gothh$ to be
\begin{align}
    \varTheta_\gothh:=
    (\phi_{\xi^\circ}\cdot\exp_\gothh)^{-1},
\end{align}
on a neighborhood of $\xi^\circ\in\calM$ such that $\phi_{\xi^\circ} \cdot \exp$ is bijective.
If the homogeneous space is reductive, then there is a natural choice of $\gothh$ \cite{cheeger1975comparison} which is the normal coordinate, however, the construction is always possible.

Given $X\in\grpG$ is an element on the symmetry group, define the following probability function:
\begin{align}
    p(\xi) = \alpha e^{-\frac{1}{2}(\varTheta_\gothh(\phi(X^{-1},\xi))^\top \Sigma^{-1}\varTheta_\gothh(\phi(X^{-1},\xi)))},
\end{align}
where $\alpha$ is a normalizing factor and $\Sigma$ is a positive-definite matrix.
Define $\epsilon$ as
\begin{align}
    \epsilon = \varTheta_\gothh(\phi(X^{-1},\xi))
\label{eq:xi_epsilon}
\end{align}
where the assumption is made that $\phi(X^{-1},\xi)$ remains in the domain of definition of the local coordinate chart.
The concentrated Gaussian distribution is obtained by assigning $\epsilon \sim \grpN(0,\Sigma)$ to a Gaussian distribution.
%That is a distribution that is Gaussian in a ball around the origin and zero outside of this.
We denote distribution induced on $\xi$ via \eqref{eq:xi_epsilon} to be the concentrated Gaussian distribution and denote this by $\xi \sim \grpN_\phi(X,0,\Sigma)$.
By construction, the expected value of $\xi$ is
$\mathrm{E}[\xi] = \phi(X,\xi^\circ)$.
Thus, for $\xi\sim\grpN_\phi(X,0,\Sigma)$ one has
$\xi = \phi_X(\varTheta_\gothh^{-1}(\epsilon))$, $\epsilon \sim \grpN(0,\Sigma)$ where $X$ is called the reference and $\epsilon$ is the local state error with mean 0 and covariance $\Sigma$.

We use the term \emph{extended concentrated Gaussian distribution} to refer to the case where the local state error distribution has non-zero mean.
For $\mu\in\R^m$ define $\xi \sim \grpN_\phi(X,\mu,\Sigma)$ to be
\begin{align}
    \xi = \phi_X(\varTheta_\gothh^{-1}(\epsilon)). \quad\epsilon \sim \grpN(\mu,\Sigma)
\end{align}
One has
\begin{align}
    \mathrm{E}[\xi]=\phi_X(\varTheta_\gothh(\mu)).
\end{align}
This extended notation allows us to explicitly model the update and reset step involving Bayesian fusion of new data and parallel transport.

\subsection{Prediction step}
Consider a discrete-time system (\ref{eq:evo}), assume there exists a state symmetry $\phi:\grpG\times\calM\rightarrow\calM$, an input symmetry $\psi:\grpG\times\vecV\rightarrow\vecV$ and an equivariant lift $\Lambda:\calM\times\vecV\rightarrow\grpG$.
% Define the observer trajectory $\hat{X}_k\in\grpG$ and the global state error $e_k\in\calM$.
Let $\xi_k\in\calM$ denote the true system state whose trajectory is determined by the external input $u_k\in\vecV$.
It can be expressed as a concentrated Gaussian distribution $\xi_k\sim \grpN_{\phi}(\hat{X}_k,0,\Sigma_k)$ on the homogeneous space.
\begin{lemma}
    Given a prior distribution $\xi_{k|k}\sim \grpN_{\phi}(\hat{X}_{k|k},0,\Sigma_{k|k})$, the prediction of $\hat{X}_{k+1|k}$ and $\Sigma_{k+1|k}$ can be best approximated by
    \begin{align}
        \hat{X}_{k+1|k} &= \hat{X}_{k|k}\Lambda(\phi_{\xi^\circ}(\hat{X}_{k|k}),u_{k+1}),\quad \hat{X}_0 = \id \\
        \Sigma_{k+1|k}&=A_{k+1}\Sigma_{k|k}A_{k+1}^\top+ P,\quad\Sigma(0)=\Sigma_0 \label{eq:sigma_predict}
    \end{align}
    where $P\in\mathbb{S}_+(m)$ is a constant state gain matrix.
\end{lemma}

\begin{proof}
The log-likelihood function of $\xi_k$ can be written as
\begin{align}
    \mathcal{L}(\xi_{k|k})=\frac{1}{2}\lVert \varTheta_\gothh(\phi(\hat{X}^{-1},\xi_{k|k}))\rVert^2_{\Sigma_{k|k}}=\frac{1}{2}\lVert \varepsilon_{k|k}\rVert^2_{\Sigma_{k|k}},
\end{align}
where $\lVert\cdot\rVert^2_\Sigma$ is the Mahalanobis norm and $\varTheta_\gothh$ is a local coordinate chart on $\calM$.
Similarly, we can derive the log-likelihood of $\xi_{k+1}$,
\begin{align}
    \mathcal{L}(\xi_{k+1|k})=\frac{1}{2}\lVert \varepsilon_{k+1|k}\rVert^2_{\Sigma_{k+1|k}}.
\end{align}
Using the linearisation of $\varepsilon_k$ (\ref{eq:linearisation}), one gets
\begin{align}
    \mathcal{L}(\xi_{k+1|k})=\frac{1}{2}\lVert A_{k+1}\varepsilon_{k|k}\rVert^2_{\Sigma_{k+1|k}}.
\end{align}
Since there is no new information added to the filter in this stage, the probability distribution should remain the same,
\begin{align}
    \mathcal{L}(\xi_{k+1|k})&= \mathcal{L}(F(\xi_{k|k})) = \mathcal{L}(\xi_{k|k});\\
    \frac{1}{2}\lVert A_{k+1}\varepsilon_{k|k}\rVert^2_{\Sigma_{k+1|k}} &= \frac{1}{2}\lVert \varepsilon_{k|k}\rVert^2_{\Sigma_{k|k}}.
\end{align}
Solving the Mahalanobis norm yields
\begin{align}
    \Sigma_{k+1|k}=A_{k+1}\Sigma_{k|k} A_{k+1}^\top,
\end{align}
as required.
\end{proof}

\subsection{Update step}
The update step involves Bayesian fusion of the predicted state estimate $\xi_{k+1|k}\sim \grpN_{\phi}(\hat{X}_{k+1|k},0,\Sigma_{k+1|k})$ with a measurement $y_{k+1} \in \R^n$ associated with a generative noise model $y_{k+1}\sim \grpN(h(\xi_{k+1}),Q)$ for $Q\in\mathbb{S}_+(n)$ the measurement covariance.
We will solve this fusion problem in local coordinates centred at the predicted state estimate $\hat{\xi}_{k+1|k} = \phi_{\hat{X}_{k+1|k}}(\mr{\xi})$.
That is, consider the predicted state distribution in local coordinates \begin{align}
    &\Pr (\varepsilon_{k+1|k})\propto \exp(-\frac{1}{2}\varepsilon_{k+1|k}^\top \Sigma_{k+1|k}^{-1}\varepsilon_{k+1|k}),\label{eq:propstate}
\end{align}
and the output likelihood $l(\varepsilon_{k+1|k}|y_{k+1}) \equiv \Pr(y_{k+1}|\xi_{k+1|k})$ is given by
\begin{align}
 \mathcal{L}(\epsilon_{k+1|k}|y_{k+1}) & \propto \exp(-\frac{1}{2}(y_{k+1}-\hat{y}_{k+1})^\top Q^{-1}(y_{k+1}-\hat{y}_{k+1})) \notag \\
& \approx \exp(-\frac{1}{2}\varepsilon_{k+1|k}^\top C_{k+1|k}^\top  Q^{-1}C_{k+1|k} \varepsilon_{k+1|k}). \label{eq:propmeas}
\end{align}
Here Equation \eqref{eq:linearoutput} is used to substitute for $\tilde{y}_{k+1} = y_{k+1} - \hat{y}_{k+1}$ and the second order error terms are ignored.
The notation $C_{k+1|k}$ indicates that the linearisation is taken at the point $\hat{\xi}_{k+1|k}$.

The probability distribution for $\Pr(\varepsilon_{k+1|k+1} | \varepsilon_{k+1|k}, y_{k+1})$, for the fused estimate $\varepsilon_{k+1|k+1}$ in local coordinates, is the Gaussian $\varepsilon_{k+1|k+1} \sim {\mathbf N}(\mu_{k+1}, \Sigma^\diamond_{k+1|k+1})$ where the mean $\mu_{k+1}$ and covariance $\Sigma^\diamond_{k+1|k+1}$ are given by the standard update formula for Gaussian fusion
\begin{align}
    \Sigma^\diamond_{k+1|k+1}&=(\Sigma_{k+1|k}^{-1}+C_{k+1|k}^\top Q^{-1} C_{k+1|k})^{-1},\\
    \mu_{k+1} &= \Sigma^\diamond_{k+1|k+1}C_{k+1|k}^\top Q^{-1}\tilde{y}_{k+1}.
\end{align}
Since this information state lives in local coordinates centred at the predicted state $\hat{\xi}_{k+1|k}$ then the updated information state can be written as an extended concentrated Gaussian
\[
\xi_{k+1|k+1} \sim \grpN_{\phi}(\hat{X}_{k+1|k},\mu_{k+1},\Sigma^\diamond_{k+1|k+1}),
\]
using the notation introduced earlier.
In particular, the update step computes the new non-zero mean $\mu_{k+1}$ without changing the reference point of the coordinates $\xi_{k+1|k}$.
This leads naturally to posing the reset step as a coordinate transformation rather than a part of the stochastic fusion process.

\subsection{Reset Step}
In the previous subsection, the update step computed the new state estimate $\hat{\xi}_{k+1|k+1}$ in local coordinates with a non-zero mean $\mu_{k+1}$.
However, the local coordinates are still centred at $\hat{\xi}_{k+1|k}$ and the next update requires a state-estimate expressed in coordinates centred at $\hat{\xi}_{k+1|k+1}$.
The reset step computes $\hat{X}_{k+1|k+1}$ and $\Sigma_{k+1|k+1}$ such that
\[
\grpN_{\phi}(\hat{X}_{k+1|k},\mu_{k+1},\Sigma^\diamond_{k+1|k+1})
\approx
\grpN_{\phi}(\hat{X}_{k+1|k+1},0,\Sigma_{k+1|k+1})
\]
where $\hat{\xi}_{k+1|k} = \phi_{\hat{X}_{k+1|k}}(\mr{\xi})$ and
$\hat{\xi}_{k+1|k+1} = \phi_{\hat{X}_{k+1|k+1}}(\mr{\xi})$.
Finding $\hat{X}_{k+1|k+1}$ is relatively straightforward.
Computing the covariance $\Sigma_{k+1|k+1}$ is more challenging and is an active topic of research \cite{mueller2017covariance}.

To compute $\hat{X}_{k+1|k+1}$ we use the mean $\mu_{k+1}$ of the local state error to derive a correction term on the Lie algebra.
% denoted by $\Delta\in\gothg$.
Let $\tD\phi_{\mr{\xi}}(\id)^\dagger$ denote a choice of right inverse of $\tD\phi_{\mr{\xi}}(\id)$; that is,
$\tD\phi_{\mr{\xi}}(\id) \tD\phi_{\mr{\xi}}(\id)^\dagger : T_{\mr{\xi}} \calM \to T_{\mr{\xi}} \calM$ is the identity map.
For $\mu_{k+1} \in\tT_{\xi^\circ}\calM$ define
\begin{align}
    \Delta_{k+1} &= \tD\phi_{\mr{\xi}}(\id)^\dagger \tD\varTheta_\gothh^{-1}\mu_{k+1}.
\end{align}
The new group element is defined to be
\begin{align}
    \hat{X}_{k+1|k+1} :=\exp(\Delta_{k+1})\hat{X}_{k+1|k}.
\end{align}
By construction $\phi_{\mr{\xi}} (\hat{X}_{k+1|k+1}) = \hat{\xi}_{k+1|k+1}$.

Left multiplication by the group element $\exp(\Delta_{k+1})$ can be interpreted as a left translation $\hat{X}_{k+1|k} \mapsto \hat{X}_{k+1|k+1} =
\tL_{\exp(\Delta_{k+1})}\hat{X}_{k+1|k}$ along the one-parameter subgroup $\exp(t \Delta_{k+1})$ for $t \in [0,1]$.
The Cartan-Schouten connections are a collection of three invariant connections for which the one-parameter subgroups are geodesics \cite{Nomizu_1954}.
The covariance $\Sigma^\diamond_{k+1|k+1}\in \mathbb{S}_+(m)$ can be thought of as a (2,0)-tensor and the inverse covariance, or information matrix, $(\Sigma^\diamond_{k+1|k+1})^{-1}$ a (0,2)-tensor.
Interpreting the curve $\exp(t \Delta_{k+1})$ as a geodesic we propose to parallel transport the covariance tensor $\Sigma^\diamond_{k+1|k+1}$ from the base point $\hat{\xi}_{k+1|k}$ to the new base point $\hat{\xi}_{k+1|k+1}$ along the curve $\exp(t \Delta_{k+1})$.
This approach provides a natural geometric manner to reset the covariance although it requires the additional structure of an affine connection that we discuss at the end of the section.
We will use the Cartan-Schouten (0)-connection, the unique torsion-free invariant connection.

Let $v\in\gothg$ be a tangent vector at identity and $\Delta \in \gothg$.
The parallel transport $\pt_{\exp(t \Delta)}(v)$ for the Cartan-Schouten (0)-connection is given by
\begin{align}
\pt_{\gamma_\Delta(t)}(v) = \tD \tL_{\exp(t\Delta)} \Ad_{\exp(-\frac{t}{2}\Delta)}(v).
\end{align}
Let $v(t)\in\gothg$ represent the left trivialisation of  $V(t)\in T_{\gamma_\Delta(t)}\grpG$; that is, $v(t):=\tD \tL_{\exp(- t\Delta)} V(t)$.
One has
\begin{align}
    v(t)=\Ad_{\exp(-\frac{t}{2}\Delta)}(v) = \exp(-\frac{t}{2}\Delta)\: v\: \exp(\frac{t}{2}\Delta).
\end{align}
Let $\Sigma$ denote the parallel transport of $\Sigma^\diamond$ over the curve $\exp(t \Delta)$.
One has
\begin{gather}
(\Sigma^\diamond)^{-1}(p_1,p_2)
= \Sigma^{-1}(\pt_{\exp(\Delta)}(p_1),\pt_{\exp(\Delta)}(p_2)); \nonumber\\
p_1^\top (\Sigma^\diamond)^{-1} p_2
= (\Ad_{\exp(-\frac{1}{2}\Delta)}(p_1))^\top \Sigma^{-1} \Ad_{\exp(-\frac{1}{2}\Delta)}(p_2); \nonumber\\
p_1^\top (\Sigma^\diamond)^{-1} p_2
= p_1^\top {\Ad_{\exp(-\frac{1}{2}\Delta)}^{\vee}}^\top \Sigma^{-1} \Ad_{\exp(-\frac{1}{2}\Delta)}^\vee(p_2),
\end{gather}
and hence
\[
\Sigma_{k+1|k+1}
= \Ad_{\exp(-\frac{1}{2}\Delta)}^{\vee} (\Sigma_{k+1|k+1}^\diamond) {\Ad_{\exp(-\frac{1}{2}\Delta)}^{\vee}}^\top.
\]
This formula corresponds to those obtained recently in the literature \cite{mueller2017covariance}.

A connection, or parallel transport is a geometric structure that is in addition to the differential geometric structure inherent in the manifold structure of the Lie-group.
However, the Cartan-Schouten connections are the only invariant connections for which one-parameter Lie-subgroups (exponential curves) are geodesics.
The (0)-connection or normal-connection is the only torsion free Cartan-Schouten connection, although it does have curvature that is expressed in the formula for parallel transport.
We believe that the torsion free property is of key importance and that the curvature is a natural consequence of the non-linear structure of the state-space.

\subsection{Summary}
Following the methedologies presented in this section, the Equivariant filter design for discrete-time system is
\begin{align}
    \intertext{Prediction:}
    \hat{X}_{k+1|k} &= \hat{X}_{k|k}\Lambda(\phi_{\xi^\circ}(\hat{X}_{k|k}),u_{k+1}),\quad \hat{X}(0) = \id\\
    \Sigma_{k+1|k}&=A_{k+1}\Sigma_{k|k}A_{k+1}^\top+ P,\quad\Sigma(0)=\Sigma_0\\
    \intertext{Update:}
    \Sigma^\diamond_{k+1|k+1}&=(\Sigma_{k+1|k}^{-1}+C_{k+1|k}^\top Q^{-1} C_{k+1|k})^{-1},\\
    \mu_{k+1} &= \Sigma_{k+1|k+1}C_{k+1|k}^\top Q^{-1}\tilde{y}_{k+1}
    \intertext{Reset:}
    \Delta_{k+1}&=\tD\phi_{\mr{\xi}}(\id)^\dagger \tD\varTheta_\gothh^{-1}\mu_{k+1},\\
    \hat{X}_{k+1|k+1}&=\exp(\Delta_{k+1})\hat{X}_{k+1|k},\label{eq:reset-Xhat}\\
    \Sigma_{k+1|k+1}&=\Ad_{\exp(-\frac{1}{2}\Delta_{k+1})}^{\vee}    \Sigma^\diamond_{k+1|k+1} {\Ad_{\exp(-\frac{1}{2}\Delta_{k+1})}^{\vee}}^\top \label{eq:reset_covariance}.
\end{align}
Note that we group the base point reset \eqref{eq:reset-Xhat} within the reset step, rather than in the update step as would be normal in a statement of the Extended Kalman Filter.
This allows us to explicitly identify the updated observer state on the group separately from the local coordinates. 
Noting that in a classical EKF the local coordinates are one-to-one with the manifold and this distinction is unnecessary.
Furthermore, in the classical EKF the underlying geometry used is flat (the local coordinates are treated as Euclidean coordinates) and the covariance update \eqref{eq:reset_covariance} is the identity.

%%%%%%%%%%%%%%%%%%%%%%%%%%%%%%%%%%%%%%%%
\section{EXAMPLE: SECOND-ORDER KINEMATICS}
\label{sec:sim}
%%%%%%%%%%%%%%%%%%%%%%%%%%%%%%%%%%%%%%%%
In this section, we consider the example of linear second-order kinematics in $\R^3$ with separate bearing and range measurements.
Range and bearing measurements are common measurement modalities in modern robotics.
These measurements could be used to reconstruct the measurement of relative position directly, however, such reconstruction distorts the uncertainties in the model, and is always better to use the raw measurements in filter design if possible.

\subsection{System Definition}
Second-order kinematics on $\R^3$ is given by
\begin{align}
    \dot{p}&=v, \notag \\
    \dot{v}&=a,
    \label{eq:example_continuous_system}
\end{align}
where the system state is $\xi=(p,v)\in\R^3\times\R^3$, and $a\in \R^3$ is the external input.
The measurement functions are given by
\begin{align}
    &y_1:=h_1(p,v)=\frac{p}{\lvert p\rvert}\in\grpS^2\\
    &y_2:=h_2(p,v)=\lvert p\rvert\in\R_+,
\end{align}
where $y_1$ and $y_2$ are bearing and range measurements respectively.
In this system, the state space is $\calM:=\R^3\times\R^3$ and output spaces are $\calN_1:=\grpS^2\subset\R^3$, $\calN_2:=\R_+$.

In this example, we extend the input space to $\vecV:=\R^3\times\R^3$.
An element in the input space can be thought of as two independent inputs $u=(\omega,a)\in\vecV$, and in the implementation we will always use $u=(0,a)$.
Note that modelling $\omega\in\R^3$ as an extra input, even though it will be set to zero for all time, is critical for the equivariance of second order kinematic systems \cite{ng2020equivariant}.
Let $t\in\R$ be a fixed time interval.
We derive the discrete system evolution function by integrating \eqref{eq:example_continuous_system} with a piecewise constant acceleration $a = a_{k+1}$,
\begin{align}
    p_{k+1}&=p_k+t(v_k+\omega_{k+1})+\frac{1}{2}t^2a_{k+1}\\
    v_{k+1}&=v_k+ta_{k+1}.
\end{align}

Define the product Lie group $\grpG:=\SO(3)\times\MR(1)\ltimes\R^3$, where $\MR$ is the group of positive reals under multiplication.
The identity element is $(\id_3,1,0)$, and the inverse element is
\begin{align}
    (R,r,\beta)^{-1}=(R^\top,\frac{1}{r},-\frac{1}{r}R^\top\beta).
\end{align}
The corresponding group multiplication is given by
\begin{align}
    (R_1,r_1,\beta_1)(R_2, r_2,\beta_2)=(R_1R_2, r_1r_2, \beta_1+r_1R_1\beta_2).
\end{align}
Define $\phi:\grpG\times\calM\rightarrow\calM$ by
\begin{align}
    \phi((R,r,\beta),(p,v_p)):=(\frac{1}{r}R^\top p, \frac{1}{r}R^\top(v-\beta)).
\end{align}
Clearly $\phi$ is a smooth, transitive right group action of $\grpG$ on $\calM$.

\subsection{Equivariant System}
Given any $(\omega,a)\in\vecV$ and $(R,r,\beta)\in\grpG$, define $\psi:\grpG\times\vecV\rightarrow\vecV$ by
\begin{align}
    \psi((R,r,\beta),(\omega,a)):=(\frac{1}{r}R^\top(\omega+\beta), \frac{1}{r}R^\top a).
\end{align}

Define the function $\Lambda:\calM\times\vecV\rightarrow\grpG$ as
\begin{align}
    \Lambda((p,v),(\omega,a))=(R_\Lambda, r_\Lambda, \beta_\Lambda).
\end{align}
Let $p'=p+t(v+\omega)+\frac{1}{2}t^2a$, then
\begin{align}
    R_\Lambda&=\id_3+(\frac{p'}{\lvert p' \rvert}\times\frac{p}{\lvert p\rvert})^\times+\frac{1-\frac{p'}{\lvert p' \rvert}^\top\frac{p}{\lvert p \rvert}}{\lvert \frac{p'}{\lvert p'\rvert}\times\frac{p}{\lvert p \rvert}\rvert^2}(\frac{p'}{\lvert p' \rvert}\times\frac{p}{\lvert p \rvert})^{\times^2},\\
    r_\Lambda&=\frac{\lvert p \rvert}{\lvert p'\rvert},\\
    \beta_\Lambda&=v-r_\Lambda R_\Lambda(v+t a),
\end{align}
where $R_\Lambda$ is the solution of $R_\Lambda\frac{p'}{\lvert p' \rvert}=\frac{p}{\lvert p \rvert}$.

\subsection{Implementation}
We simulate an oscillatory trajectory for second-order kinematics to verify the performance of the proposed filter.
The state is initialized with $(p,v)=((0,0,50),(0,0,0))\in\R^3\times\R^3$, with accelaration $a=(0,\cos(\tau),0)$.
% The time interval for the discrete-time system is $dt=0.01s$.
The trajectory is realized using Euler integration at time step $t=10^{-4}s$, sampled at 100 Hz.
The estimator has an acceleration sensor that reads $a_k\in\R^3$ but corrupted with piecewise constant zero-mean white Gaussian noise with variance $0.05 m/s^2$ per axis.
Additionally, the sensor provides bearing and range measurements and is also corrupted with Gaussian noises $\mu_b\sim\grpN(0,1^2)$ and $\mu_r\sim\grpN(0,1^2)$ in degrees and metres respectively.
We implement the extended Kalman filter with linearised output, the continuous-time EqF \cite{vanGoor_Hamel_Mahony_2021} and the discrete-time EqF.
For the discrete-time EqF, we follow the design procedure presented in Section \ref{subsec:eqf}.
The local coordinate chart for the state is the normal coordinates on $\calM$ derived from projecting exponentials from the Lie group to the manifold through the group action.
The gain matrices are chosen based on the true noise parameters.
All filters are initialized at $(p,v)=((0,0,50),(0,0,0))$ with Gaussian noises $\mu_p\sim\grpN(0,7.5^2)$ and $\mu_v\sim\grpN(0,2^2)$ in meters and meters per second respectively.
The filters are running at 100 Hz.

To compare their performance, we plot the error in position and velocity, as well as the filter energy.
The filter energy is $\frac{1}{m}\varepsilon^\top \Sigma^{-1} \varepsilon$ where $m$ is the dimension of the state.
The expected value of filter energy should be 1, while smaller or larger energy indicates the filter is either under-confident or over-confident about the estimation.
The first comparison is between the discrete-time EqF with and without reset covariance after the update, to demonstrate the effectiveness of the proposed method.
The second comparison is among the conventional EKF, continuous-time EqF and the proposed discrete-time EqF, to show the advantage of the discrete-time filter in this example.
\subsection{Simulation Results}

\begin{figure}
    \centering
    \includegraphics[width=0.7\linewidth]{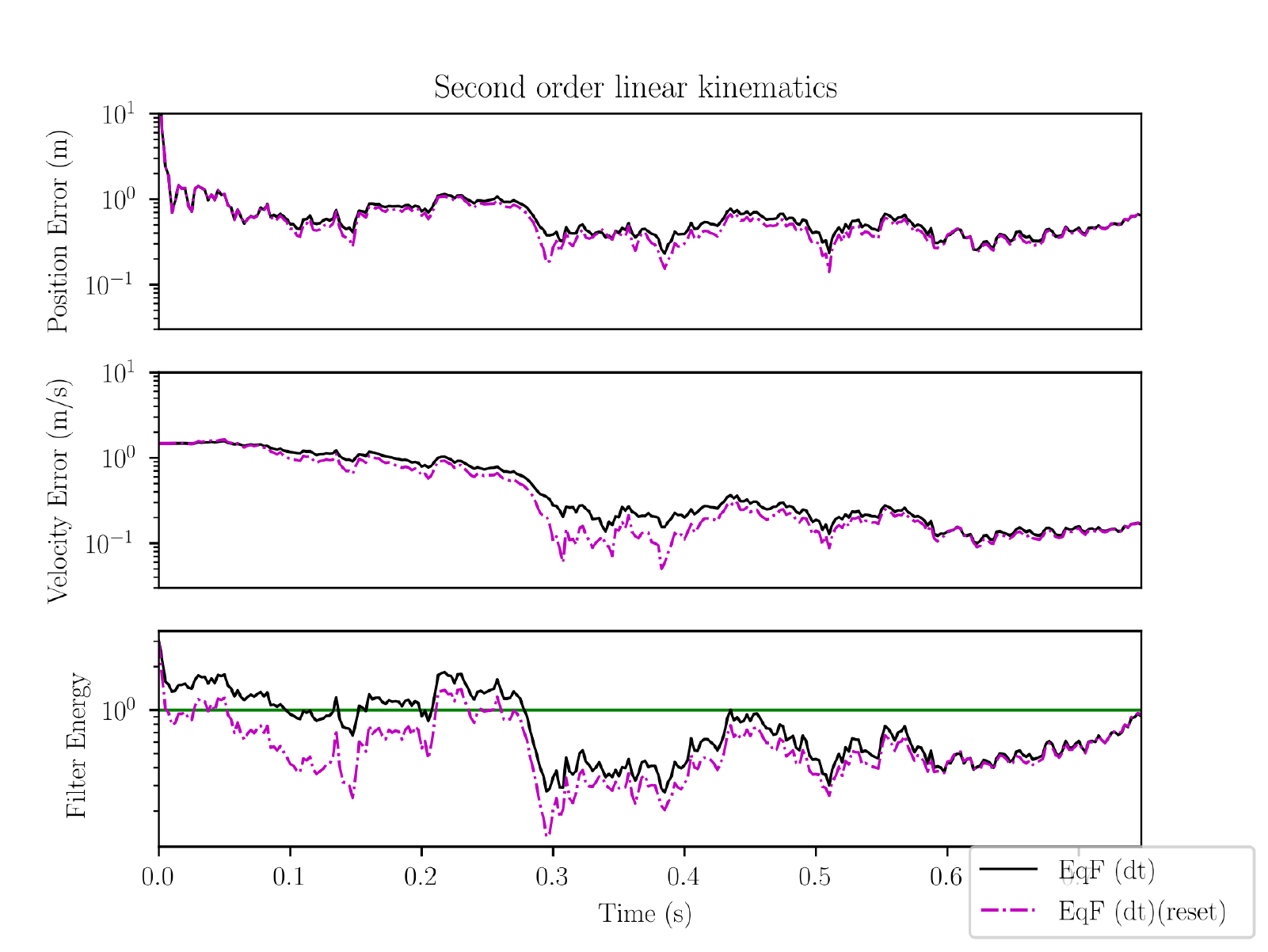}
    \caption{Comparison of discrete-time EqF implementation with/without reset step in the transient stage. The black solid line shows the discrete-time EqF without reset. The purple dashed line shows the discrete-time EqF with reset. The green horizontal line in the third subplot is for filter energy = 1. }
    \label{fig:reset}
\end{figure}

Figure \ref{fig:reset} shows the performance of the discrete-time EqF with (black, solid) and without (purple, dash-dot) the reset step.
The EqF with reset step clearly outperforms the same filter without the reset step.
The improvement occurs during the transient when the correction term in the filter is larger and the effect of parallel transport is more significant.
Later in the plot ($>0.7s$) the filters have converged and the effect of the reset step is less noticeable (Fig.~\ref{fig:ct_dt}).
The position error is less significant than the velocity error although still appreciable.
The likely reason for this lies in the fact that the position is directly observable from range and bearing measurements, while the velocity estimation depends on correlation of error encoded in the covariance estimate.

Figure \ref{fig:ct_dt} shows the performance of EKF (red, solid), discretized continuous-time EqF (blue, dash) and discrete-time EqF (black, solid) over the same trajectory.
The continuous-time and discrete-time EqFs (both with reset modification of the Riccati) have very similar performance during the transient, while EKF (without reset modification) takes much longer to converge.
The discrete-time EqF has better asymptotic performance compared to the discretized continuous-time EqF, especially in the velocity error.
We hypothesise that this performance gain is due to the fact that the error incurred in a discrete state update implemented as a symmetry transform is better adapted to the discrete equivariant filter equations than the error incurred by the discretization of a continuous-time filter, even when the associated filter was based on equivariant design principles in continuous time.

\begin{figure}
    \centering
    \includegraphics[width=0.7\linewidth]{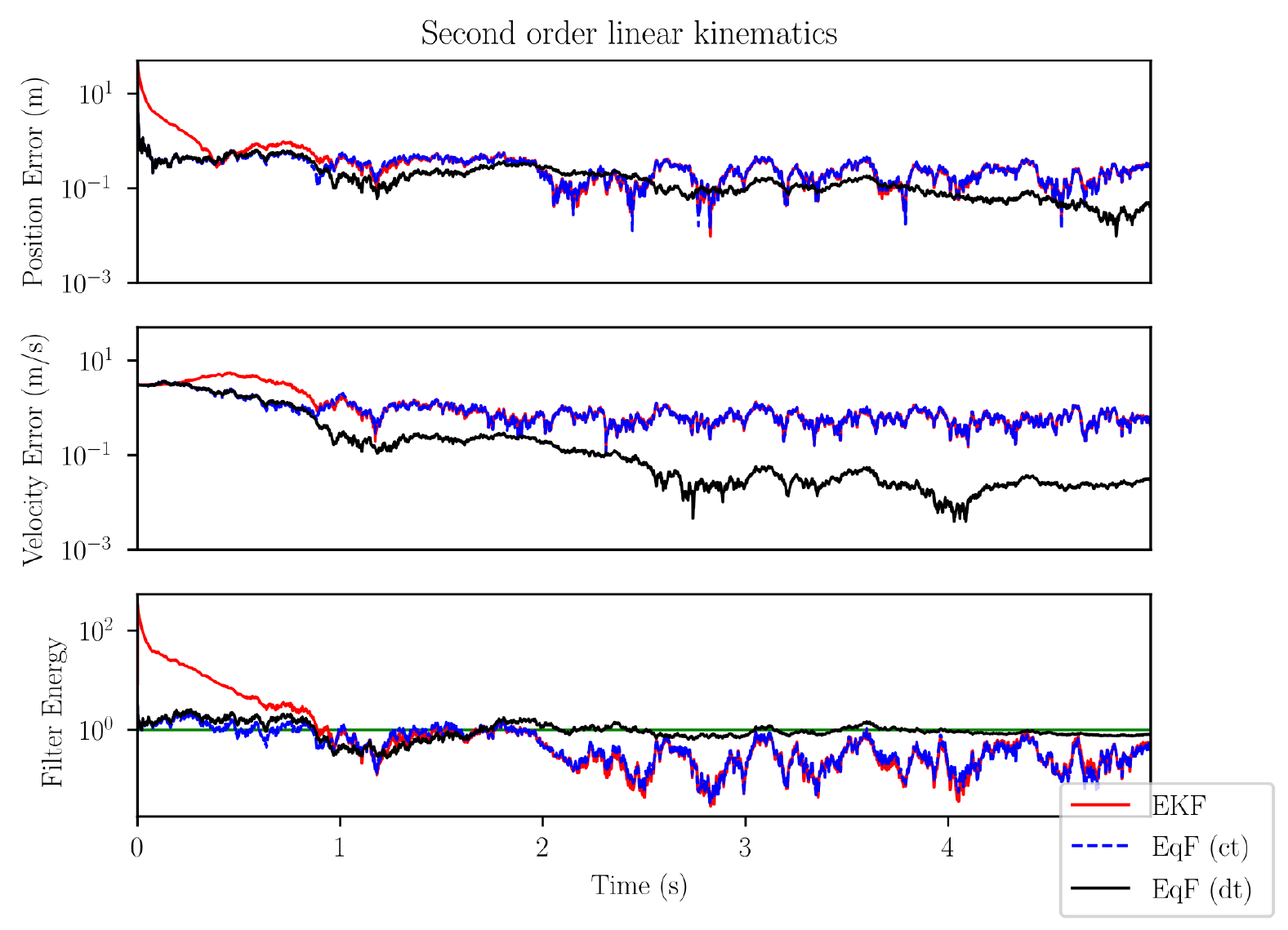}
    \caption{Comparison of three different filters for second order kinematics. The red solid line shows the Extended Kalman filter with reconstructed position measurements. The blue dashed line shows the EqF with polar symmetry implemented in continuous time. The black solid line shows the discrete-time EqF that we propose. The green horizontal line in the third subplot is for filter energy = 1.}
    \label{fig:ct_dt}
\end{figure}

%%%%%%%%%%%%%%%%%%%%%%%%%%%%%%%%%%%%%%%%
\section{CONCLUSIONS}
\label{sec:conclusion}
%%%%%%%%%%%%%%%%%%%%%%%%%%%%%%%%%%%%%%%%
This paper presents an equivariant filter design methodology for discrete-time systems on homogeneous spaces.
The algorithm contains an extra reset step which explicitly computes the parallel transport of covariance matrix using the invariant affine connection.
The example of second order kinematics system with range and bearing measurements is used to detail the implementation.
The simulation demonstrates the convergence of discrete-time EqF as well as the improvement from the covariance reset step.
It is shown that the discrete system lift brings better asymptotic performance of the filter.

\bibliography{arxiv_version}

\begin{thebibliography}{10}
\providecommand{\url}[1]{#1}
\csname url@samestyle\endcsname
\providecommand{\newblock}{\relax}
\providecommand{\bibinfo}[2]{#2}
\providecommand{\BIBentrySTDinterwordspacing}{\spaceskip=0pt\relax}
\providecommand{\BIBentryALTinterwordstretchfactor}{4}
\providecommand{\BIBentryALTinterwordspacing}{\spaceskip=\fontdimen2\font plus
\BIBentryALTinterwordstretchfactor\fontdimen3\font minus
  \fontdimen4\font\relax}
\providecommand{\BIBforeignlanguage}[2]{{%
\expandafter\ifx\csname l@#1\endcsname\relax
\typeout{** WARNING: IEEEtran.bst: No hyphenation pattern has been}%
\typeout{** loaded for the language `#1'. Using the pattern for}%
\typeout{** the default language instead.}%
\else
\language=\csname l@#1\endcsname
\fi
#2}}
\providecommand{\BIBdecl}{\relax}
\BIBdecl

\bibitem{FB_ADL_04}
F.~Bullo and A.~D. Lewis, \emph{Geometric Control of Mechanical Systems}, ser.
  Texts in Applied Mathematics.\hskip 1em plus 0.5em minus 0.4em\relax New
  York-Heidelberg-Berlin: Springer Verlag, 2004, vol.~49.

\bibitem{Jurdjevic_1997}
V.~Jurdjevic, \emph{\BIBforeignlanguage{English}{Geometric control
  theory}}.\hskip 1em plus 0.5em minus 0.4em\relax Cambridge;New York, NY,
  USA;: Cambridge University Press, 1997, vol. 52.

\bibitem{Salcudean_1991}
S.~Salcudean, ``A globally convergent angular velocity observer for rigid body
  motion,'' \emph{IEEE Transactions on Automatic Control}, vol.~36, no.~12, pp.
  1493--1497, 1991.

\bibitem{Thienel_2003}
J.~Thienel and R.~Sanner, ``A coupled nonlinear spacecraft attitude controller
  and observer with an unknown constant gyro bias and gyro noise,'' \emph{IEEE
  Transactions on Automatic Control}, vol.~48, no.~11, pp. 2011--2015, 2003.

\bibitem{Aghannan_2003}
N.~Aghannan and P.~Rouchon, ``An intrinsic observer for a class of lagrangian
  systems,'' \emph{IEEE Transactions on Automatic Control}, vol.~48, no.~6, pp.
  936--945, 2003.

\bibitem{Mahony_2008}
R.~Mahony, T.~Hamel, and J.-M. Pflimlin, ``Nonlinear complementary filters on
  the special orthogonal group,'' \emph{IEEE Transactions on Automatic
  Control}, vol.~53, no.~5, pp. 1203--1218, 2008.

\bibitem{Bonnabel_2007}
S.~Bonnabel, ``Left-invariant extended kalman filter and attitude estimation,''
  in \emph{2007 46th IEEE Conference on Decision and Control}, 2007, pp.
  1027--1032.

\bibitem{Bonnabel_2009}
S.~Bonnable, P.~Martin, and E.~Salaün, ``Invariant extended kalman filter:
  theory and application to a velocity-aided attitude estimation problem,'' in
  \emph{Proceedings of the 48h IEEE Conference on Decision and Control (CDC)
  held jointly with 2009 28th Chinese Control Conference}, 2009, pp.
  1297--1304.

\bibitem{Barrau_2017}
A.~Barrau and S.~Bonnabel, ``The invariant extended kalman filter as a stable
  observer,'' \emph{IEEE Transactions on Automatic Control}, vol.~62, no.~4,
  pp. 1797--1812, 2017.

\bibitem{vanGoor_Hamel_Mahony_2021}
P.~van Goor, T.~Hamel, and R.~Mahony, ``Equivariant filter (eqf): A general
  filter design for systems on homogeneous spaces,'' in \emph{2020 59th IEEE
  Conference on Decision and Control (CDC)}, 2020, pp. 5401--5408.

\bibitem{mahony2020equivariant}
R.~Mahony, T.~Hamel, and J.~Trumpf, ``Equivariant systems theory and observer
  design,'' \emph{arXiv preprint arXiv:2006.08276}, 2020.

\bibitem{Chung_1990}
\BIBentryALTinterwordspacing
S.-T. Chung, ``Digital aspects of nonlinear synthesis problems,'' Ph.D.
  dissertation, 1990. [Online]. Available:
  \url{https://www.proquest.com/dissertations-theses/digital-aspects-nonlinear-synthesis-problems/docview/303843416/se-2?accountid=8330}
\BIBentrySTDinterwordspacing

\bibitem{crouch1993numerical}
P.~E. Crouch and R.~Grossman, ``Numerical integration of ordinary differential
  equations on manifolds,'' \emph{Journal of Nonlinear Science}, vol.~3, no.~1,
  pp. 1--33, 1993.

\bibitem{hairer2006geometric}
E.~Hairer, M.~Hochbruck, A.~Iserles, and C.~Lubich, ``Geometric numerical
  integration,'' \emph{Oberwolfach Reports}, vol.~3, no.~1, pp. 805--882, 2006.

\bibitem{Hori_1989}
N.~Hori, T.~Mori, and P.~N. Nikiforuk, ``A new perspective for discrete-time
  models of a continuous-time system,'' in \emph{1989 American Control
  Conference}, 1989, pp. 1659--1664.

\bibitem{maybeck1982stochastic}
P.~S. Maybeck, \emph{Stochastic models, estimation, and control}.\hskip 1em
  plus 0.5em minus 0.4em\relax Academic press, 1982.

\bibitem{Bourmaud_2015}
G.~Bourmaud, R.~Mégret, M.~Arnaudon, and A.~Giremus, ``Continuous-discrete
  extended kalman filter on matrix lie groups using concentrated gaussian
  distributions,'' \emph{Journal of Mathematical Imaging and Vision}, vol.~51,
  no.~1, p. 209–228, Jan 2015.

\bibitem{barrau2018invariant}
A.~Barrau and S.~Bonnabel, ``Invariant kalman filtering,'' \emph{Annual Review
  of Control, Robotics, and Autonomous Systems}, vol.~1, pp. 237--257, 2018.

\bibitem{forster2016manifold}
C.~Forster, L.~Carlone, F.~Dellaert, and D.~Scaramuzza, ``On-manifold
  preintegration for real-time visual--inertial odometry,'' \emph{IEEE
  Transactions on Robotics}, vol.~33, no.~1, pp. 1--21, 2016.

\bibitem{barrau2019linear}
A.~Barrau and S.~Bonnabel, ``Linear observed systems on groups,'' \emph{Systems
  \& Control Letters}, vol. 129, pp. 36--42, 2019.

\bibitem{doi:10.1146/annurev-control-061520-010324}
\BIBentryALTinterwordspacing
R.~Mahony, P.~van Goor, and T.~Hamel, ``Observer design for nonlinear systems
  with equivariance,'' \emph{Annual Review of Control, Robotics, and Autonomous
  Systems}, vol.~5, no.~1, p. null, 2022. [Online]. Available:
  \url{https://doi.org/10.1146/annurev-control-061520-010324}
\BIBentrySTDinterwordspacing

\bibitem{markley2003attitude}
F.~L. Markley, ``Attitude error representations for kalman filtering,''
  \emph{Journal of guidance, control, and dynamics}, vol.~26, no.~2, pp.
  311--317, 2003.

\bibitem{mueller2017covariance}
M.~W. Mueller, M.~Hehn, and R.~D’Andrea, ``Covariance correction step for
  kalman filtering with an attitude,'' \emph{Journal of Guidance, Control, and
  Dynamics}, vol.~40, no.~9, pp. 2301--2306, 2017.

\bibitem{wolfe2011bayesian}
K.~C. Wolfe and M.~Mashner, ``Bayesian fusion on lie groups,'' \emph{Journal of
  Algebraic Statistics}, vol.~2, no.~1, 2011.

\bibitem{cheeger1975comparison}
J.~Cheeger, D.~G. Ebin, and D.~G. Ebin, \emph{Comparison theorems in Riemannian
  geometry}.\hskip 1em plus 0.5em minus 0.4em\relax North-Holland Amsterdam,
  1975, vol.~9.

\bibitem{Nomizu_1954}
\BIBentryALTinterwordspacing
K.~Nomizu, ``Invariant affine connections on homogeneous spaces,''
  \emph{American Journal of Mathematics}, vol.~76, no.~1, pp. 33--65, 1954.
  [Online]. Available: \url{http://www.jstor.org/stable/2372398}
\BIBentrySTDinterwordspacing

\bibitem{ng2020equivariant}
Y.~Ng, P.~van Goor, T.~Hamel, and R.~Mahony, ``Equivariant systems theory and
  observer design for second order kinematic systems on matrix lie groups,'' in
  \emph{2020 59th IEEE conference on decision and control (CDC)}.\hskip 1em
  plus 0.5em minus 0.4em\relax IEEE, 2020, pp. 4194--4199.

\end{thebibliography}
\bibliographystyle{IEEEtran}

\end{document}